\documentclass[11pt]{article}
\usepackage[utf8]{inputenc}
\usepackage[margin=1in]{geometry}

\usepackage{graphicx}
\usepackage{xcolor}
\usepackage{amsmath}
\usepackage{amssymb}
\usepackage{amsthm}
\usepackage{algorithm}
\usepackage{algpseudocode}
\usepackage{hyperref}
\hypersetup{
     colorlinks = true,
     linkcolor = blue,
     citecolor = blue,
     filecolor = blue,
     urlcolor = blue
     }
\usepackage{cleveref}
\usepackage{framed}
\usepackage{bbm}

\algnotext{EndFor}
\algnotext{EndIf}
\algnotext{EndWhile}
\algnotext{EndProcedure}

\allowdisplaybreaks

\newtheorem{observation}{Observation}
\newtheorem{lemma}{Lemma}
\newtheorem{theorem}[lemma]{Theorem}

\newtheorem{definition}[lemma]{Definition}

\newtheorem{assumption}[lemma]{Assumption}

\newcommand{\cC}{\mathcal{C}}

\newcommand{\cH}{\mathcal{H}}

\newcommand{\N}{\mathbb{N}}

\newcommand{\bone}{\mathbf{1}}

\newcommand{\TS}{\mathrm{TS}}
\newcommand{\TD}{\mathrm{TD}}
\newcommand{\RTD}{\mathrm{RTD}}

\newcommand{\ind}{\mathbbm{1}}

\definecolor{Gred}{RGB}{219, 50, 54}
\definecolor{Ggreen}{RGB}{60, 186, 84}
\definecolor{Gblue}{RGB}{72, 133, 237}
\definecolor{Gyellow}{RGB}{247, 178, 16}
\definecolor{ToCgreen}{RGB}{0, 128, 0}
\definecolor{myGold}{RGB}{231,141,20}
\definecolor{myBlue}{rgb}{0.19,0.41,.65}
\definecolor{myPurple}{RGB}{175,0,124}

\usepackage[colorinlistoftodos,prependcaption,textsize=scriptsize]{todonotes}
\providecommand{\Comments}{3}
\ifnum\Comments=2
\paperwidth=\dimexpr \paperwidth + 2.8cm\relax
\evensidemargin=\dimexpr\evensidemargin + 2.2cm\relax
\fi
\newcommand{\mytodo}[1]{\ifnum\Comments=1{#1}\fi}

\newcommand{\tableoftodos}{\ifnum\Comments=1 \listoftodos[Comments/To Do's] \fi}

\title{A Note on Hardness of Computing Recursive Teaching Dimension}

\author{
\makebox[.3\linewidth]
{Pasin Manurangsi}\\
Google Research\\
Bangkok, Thailand\\
\texttt{pasin@google.com}
}

\date{\today}
\begin{document}

\maketitle

\begin{abstract}
In this short note, we show that the problem of computing the recursive teaching dimension (RTD) for a concept class (given explicitly as input) requires $n^{\Omega(\log n)}$-time, assuming the exponential time hypothesis (ETH). This matches the running time $n^{O(\log n)}$ of the brute-force algorithm for the problem.
\end{abstract}

\section{Introduction}

Given a finite domain set $X$, a \emph{concept} is a boolean function $c: X \to \{0, 1\}$. A \emph{concept class} $\cC$ is simply a set of concepts. We write $c|_S$ to denote the restriction of a concept $c$ on a subset $S \subseteq X$, i.e., $c|_S: S \to \{0, 1\}$ such that $c|_S(x) = c(x)$ for all $x \in S$. 

Theory of learning explores how to learn a concept class in many different settings. One such setting is when there is a teacher who wishes to teach a concept $c \in \cC$ to a learner, where the teaching is done through giving examples $(x, c(x))$. The notions of teaching set and teaching dimension--formalized below--captures the sample complexity required to teach in this model.

\begin{definition}[Teaching Set and Teaching Dimension~\cite{Shinohara91,GoldmanK95}]
A \emph{teaching set} of a concept $c$ with respect to a concept class $\cC$ is a subset $S$ such that $c|_S \ne c'|_S$ for any $c' \in \cC \setminus \{c\}$. We write $\TS(c, \cC)$ to denote the minimum size teaching set of a concept $c$ w.r.t. a concept class $\cC$.

The \emph{teaching dimension} of a concept class $\cC$, denoted by $\TD(\cC)$, is defined as $\max_{c \in \cC} \TS(c, \cC)$.
\end{definition}

Although this is a compelling and natural model of teaching, the teaching dimension of many simple concept classes turn out to be prohibitively large. For example, the concept class consisting of the all-zero function and point functions\footnote{Point functions are those that evaluate to one on a single input.}--despite having VC dimension of only two, have teaching dimension that is as large as the domain size. This specific example demonstrates that having a single ``hard to teach'' function (i.e. the all-zero function) can significantly increase the teaching dimension. In an effort to make the notion more robust, a more relaxed model has been proposed~\cite{ZillesLHZ11,DoliwaFSZ14,SameiSYZ14}. In this model, teaching happens in ``layers''. In the first layer / iteration, the teacher may only teach the concepts with smallest teaching dimension. These are then removed from the concept class and the process continue. This process can be formalized as a teaching plan and the associated sample complexity is referred to as the \emph{recursive teaching dimension (RTD)}:

\begin{definition}[Recursive Teaching Dimension~\cite{ZillesLHZ11}]
For a concept class $\cC$, a \emph{teaching plan} is a sequence $(c_1, S_1), \dots, (c_m, S_m)$ such that $\cC = \{c_1, \dots, c_m\}$ and, for all $i \in [m]$, $S_i$ is a teaching set of the concept $c_i$ with respect to the concept class $\{c_i, \dots, c_m\}$.

The \emph{recursive teaching dimension} of a concept class $\cC$, denoted by $\RTD(\cC)$, is the minimum of $\max_{i \in [m]} S_i$ across all teaching plans of the class $\cC$.
\end{definition}

It is simple to see that RTD of the aforementioned concept class is small: the teacher can just teach the point functions first (requiring just a single example each) and, once they are removed, the remaining concept must be the all-zero function (requiring zero example to teach). It turns out that this is not an isolated phenomenon: Improving on and extending several known bounds~\cite{Kuhlmann99,DoliwaFSZ14,MoranSWY15,ChenCCT16,WigdersonY16}, Hu et al.~\cite{HuWLW17} showed that the RTD of a concept class is at most a quadratic function of its VC Dimension. Given the small sample complexity, RTD emerges as a compelling model for teaching.



\subsection{Our contributions}

Although computing the minimum teaching set size of a given concept w.r.t. a given concept class is known to be NP-hard~\cite{Shinohara91}, to the best of our knowledge, there was no prior study of the computational complexity of the recursive teaching dimension (although the problem was suggested in \cite{ManurangsiR17}). Our main contribution is to prove that the latter is also unlikely to be solvable in polynomial time.

To formalize our results, we write $k$-RTD to denote the problem of determining whether a given concept class $\cC$ has recursive teaching dimension at most $k$. Here $\cC$ is given as input in explicit form, i.e. as a boolean matrix $M \in \{0, 1\}^{\cC \times X}$ where $c(x) = M_{c, x}$.

Let us note that $k$-RTD can be solved in $n^{O(\log n)}$ where $n$ denote the input size (i.e. $|\cC| \times |X|$). This is because\footnote{This follows e.g. from Corollary 8 and Equation (1) of \cite{DoliwaFSZ14}.} $\RTD(\cC) \leq \lceil \log |\cC| \rceil$. Thus, if $k > \log |\cC|$, we can simply return YES. Otherwise, for each concept, we may enumerate all $k$-size subsets of the domain $X$ to determine whether it is a teaching set in time $n^{O(k)} \leq n^{O(\log n)}$. We then strip out a concept with teaching dimension at most $k$ and repeat the process. We either fail at some point (corresponding to the NO case) or we find a teaching plan whose teaching sets all have size at most $k$ (corresponding to the YES case).

Given that $k$-RTD can be solved in $n^{O(\log n)}$ time, it is unlikely to be NP-hard; otherwise, all problems in NP would be solvable in $n^{O(\log n)}$ time. To this end, we need a different complexity assumption to prove the hardness of our problem. Our result relies on the Exponential Time Hypothesis (ETH)~\cite{ImpagliazzoP01,ImpagliazzoPZ01}\footnote{See \Cref{sec:eth} for a formal definition and discussion on ETH.}. Under this hypothesis, we show that $k$-RTD requires $n^{\Omega(\log n)}$ time, as stated below. In other words, our lower bound asserts, under ETH, that the aforementioned brute-force algorithm is essentially the best possible for $k$-RTD.

\begin{theorem} \label{thm:main-lb}
Assuming ETH, there is no $n^{o(\log n)}$-time algorithm for $k$-RTD.
\end{theorem}

\section{Preliminaries}
\label{sec:prelim}

We let $\bone_X$ denote the all-one function on domain $X$; when the domain is clear from context, we may omit the subscript $X$. Let $\ind$ denote the indicator variable (i.e. $\ind[\psi] = 1$ if the condition $\psi$ holds and $\ind[\psi] = 0$ otherwise). 

In our reduction, it will be convenient to deal with tuples and set of tuples. Recall that, for two sets $A, B$, we write $A \times B$ to denote $\{(a, b) \mid a \in A, b \in B\}$. Furthermore, we let $\pi_1: A \times B \to A$ (resp. $\pi_2: A \times B \to B$) denote the projection onto the first (resp. second) coordinate, i.e. $\pi_1((a, b)) = a$ (resp. $\pi_2((a, b)) = b$). Finally, when $A, B$ have the same size, we use $\nu(A, B)$ to denote any subset of $A \times B$ such that $\pi_1(\nu(A, B)) = A$ and $\pi_2(\nu(A, B)) = B$.

\subsection{Minimum Teaching Set Sizes and Recursive Teaching Dimension}

It will be convenient to use a characterization of RTD in terms of the minimum teaching set sizes among subsets of concepts. To do this, let us first define $\TD_{\min}(\cC) := \min_{c \in \cC} \TS(c, \cC)$. The relation between this quantity and RTD is given below:

\begin{lemma}[\cite{DoliwaFSZ14}] \label{lem:rtd-characterization}
For any concept class $\cC$, $\RTD(\cC) = \max_{\cC' \subseteq \cC} \TD_{\min}(\cC')$.
\end{lemma}

\subsection{Exponential Time Hypothesis}
\label{sec:eth}

Recall that in the 3SAT problem, we are given a 3CNF formula with $n$ variables and $m$ clauses. The goal is to decide whether there is an assignment that satisfies all the clauses.

3SAT is one of the canonical NP-complete problems, included in Karp's seminal paper~\cite{Karp72}, and is often used as the starting point in NP-hardness reductions. The brute-force algorithm for 3SAT is to enumerate all $2^n$ assignments and check whether each of them satisfies all the clauses. This results in a running time of $2^{O(n)}$. Despite significant efforts from researchers (see e.g.~\cite{HansenKZZ19} and references therein), no $2^{o(n)}$-time algorithm is known. This led to a conjecture--known as the \emph{Exponential Time Hypothesis (ETH)}--that such an algorithm does not exist.

\begin{assumption}[Exponential Time Hypothesis (ETH)~\cite{ImpagliazzoP01,ImpagliazzoPZ01}]
3SAT cannot be solved in $2^{o(n)}$ time where $n$ denotes the number of variables.
\end{assumption}

Supporting evidences of this conjecture include lower bounds against certain (strong) algorithmic techniques, such as the sum-of-squares (aka Lasserre) hierarchy~\cite{Schoenebeck08}.

Since its proposal, ETH has been shown to provide tight running time lower bounds for numerous fundamental NP-hard problems and parameterized problems (see e.g.~\cite{LokshtanovMS11}). More recently, it has also been used as a basis for lower bounds for quasi-polynomial problems~\cite{BravermanKW15,BravermanKRW17,Rubinstein17,Manurangsi17,Rubinstein17b,ManurangsiR17,Man22}; we continue this line of research--by employing it for our lower bound for RTD.

\subsection{Dominating Set}

Our reduction will be from the Dominating Set problem.  Recall that, in an undirected graph $G = (V, E)$, we say that a vertex $u$ \emph{dominates} another vertex $v$ iff either $u = v$ or there is an edge between $u$ and $v$. We say that a set $T \subseteq V$ is a \emph{dominating set} iff every vertex in $V$ is dominated by at least one vertex in $T$.

\begin{definition}[$k$-Dominating Set]
In the $k$-Dominating Set problem, we are given as input an $N$-vertex graph $G$. The goal is to decide whether there exists a dominating set of $G$ of size $k$.
\end{definition}

An obvious algorithm is to enumerate through all $k$-subsets of vertices and check whether each of them is a dominating set; this runs in time $N^{O(k)}$. It is known that, under ETH, this is essentially the best possible running time, even for $k = O(\log N)$:

\begin{theorem} \label{thm:log-domset-lb}
Assuming ETH, there is no $N^{o(\log N)}$-time algorithm for $O(\log N)$-Dominating Set.
\end{theorem}

The above result is folklore and is implicitly implied by a reduction dating back to~\cite{MegiddoV88}, but we are not aware of an explicit reference so we give a proof for completeness in \Cref{sec:proof-domset}.

\section{Main Reduction}

As alluded to earlier, our main reduction will be from the $k$-Dominating Set problem to $k$-RTD:

\begin{theorem} \label{thm:main-red}
There is a $2^{O(k)} N^{O(1)}$-time reduction from $k$-Dominating Set to $k$-RTD.
\end{theorem}

Before we prove \Cref{thm:main-red}, let us note that it easily implies our main theorem (\Cref{thm:main-lb}).

\begin{proof}[Proof of \Cref{thm:main-lb}]
Suppose that there is an algorithm $A$ for $k$-RTD that runs in time $n^{o(\log n)}$. Given any $k$-Dominating Set instance $(G, k)$ where $k = O(\log N)$, we use the reduction in \Cref{thm:main-red} to produce an instance $(\cC, k)$ of $k$-RTD and run $A$ on $\cC$. This algorithm solves $k$-Dominating Set problem. Furthermore, since the reduction runs in time $2^{O(k)}N^{O(1)} = N^{O(1)}$, we have $n = N^{O(1)}$ and thus this algorithm runs in time $N^{o(\log N)}$. From \Cref{thm:log-domset-lb}, this violates ETH.
\end{proof}

We now give a high-level overview of our reduction used in \Cref{thm:main-red}. We will sometimes be informal; the full formal proof will be presented later.

\paragraph{Overview I: First Attempt.} Our reduction starts from the proof of \cite{Shinohara91} who show that computing $\TS(c^*, \cC)$ for a given concept class $\cC$ and a given concept $c^* \in \cC$ is NP-hard\footnote{The original reduction is from the Hitting Set problem, which is equivalent to the Dominating Set problem. We formulate it in terms of the latter for consistency.}. The reduction is as follows. Starting from a $k$-Dominating Set instance $G = (V, E)$, let the domain $X$ be $V$. For each $u \in V$, create a ``vertex'' concept $c_u$ such that $c_u(v) = 0$ iff $v$ dominates $u$. Then, create another ``constraint'' concept $c^* = \bone_X$, i.e. the all-one function. It is simple to check that a set $T \subseteq X$ is a teaching set for $c^*$ if and only if it is a dominating set of $G$. Therefore, $\TS(c^*, \cC) \leq k$ iff $G$ is a YES instance of $k$-Dominating Set.

While this reduction gives the hardness for computing $\TS(c^*, \cC)$, it does not give the hardness for computing $\RTD(\cC)$. The reason is that we do not have any guarantees on the minimum teaching set sizes of the vertex concepts $c_u$. If these are all smaller than $k$, then it could be the case that $\RTD(\cC) < k$ regardless of whether there is a $k$-size dominating set. (That is, we can simply first teach the vertex concepts.)

\paragraph{Overview II: Handling Vertex Concepts.} To solve this issue, instead of having a single vertex concept $c_u$ for each $u \in V$, we will replicate the concepts multiple times. Specifically, we extend the domain $X$ to be $V \cup Y$ where $Y$ is carefully chosen. Each copy $c_{u, i}$ has the same values as before on $V$ but they have different values on $Y$. Finally, we again let $c^*$ evaluates to one on $Y$. The key here is to choose $c_{u, i}(Y)$ carefully in such a way that the minimum teaching set for $c_{u, i}$ w.r.t. $\{c_{u, j}\}_{j}$, i.e. the concept class of all vertex concepts corresponding to the same vertex, is $k$. Furthermore, it should become $k + 1$ if we add the all-one function into the concept class. Roughly speaking, this means that $c^*$ has to be taught first before the vertex concepts.

This ostensibly solves our aforementioned problem, but 
unfortunately it introduces a different issue: Such a construction in turn makes $\TS(c^*, \cC)$ small. In other words, we can teach $c^*$ using a teaching set of size less than $k$ regardless of whether there is a small dominating set. This violates the soundness of the reduction. 

\paragraph{Overview III: Handling Constraint Concept.}
Our final modification is to fix this problem. It turns out that the fix is quite similar to before: Replicate $c^*$ multiple times. While this high-level idea is simple, making it work is delicate. Specifically, the issue here is that, if we replicate $c^*$ multiple times, then differentiating the different copies requires us to set $c^*(x) = 0$ for some $x \in X$. Now, picking such an $x$ in a teaching set may allow us to get away without picking a dominating set. To handle this, we have to replicate the vertex set $V$ multiple times in the domain $X$ and set the values of each $c^*$ carefully in such a way that any $k$-size teaching set must select only the $x \in X$ such that $c^*(x) = 1$. This completes the high-level overview of our approach.

\subsection{Proof of \Cref{thm:main-red}}

We will now proceed to the formal proof. In the above overview, we said that the replication has to be done using carefully chosen patterns of values. Such patterns are encapsulated in the following gadget, which will be used for replicating both the vertex concepts and constraint concepts.

\begin{lemma}[Gadget] \label{lem:gadget}
For any $k \in \N$, let $p = 2k + 1$ and $q = \binom{2k + 1}{k}$. Then, there exists a $q$-size concept class $\cH$ over a $p$-size domain $Z$ such that the following conditions hold.
\begin{enumerate}
\item For every $h \in \cH$, $\TS(h, \cH) = k$.
\item For every $h \in \cH$, if $S$ is a $k$-size teaching set of $h$ w.r.t. $\cH$, then $h|_S = \bone$.
\item For every $h \in \cH$, $\TS(h, \cH \cup \{\bone\}) \geq k + 1$.
\end{enumerate}
Moreover, the class $\cH$ can be constructed in $(pq)^{O(1)}$ time.
\end{lemma}

\begin{proof}
Let $\cH$ be the class of all functions $h: Z \to \{0, 1\}$ such that $|h^{-1}(1)| = k$. It is clear that $|\cH| = q$ and that it can be constructed in $(pq)^{O(1)}$ time. Next, we prove the three properties:
\begin{enumerate}
\item Consider any $h \in \cH$. It is simple to see that $h^{-1}(1)$ is a teaching set of $h$ and, by definition of $\cH$, it has size $k$. Thus, we have $\TS(h, \cH) \leq k$.

To show that $\TS(h, \cH) \geq k$, consider any $S \subseteq Z$ of size $k - 1$. Note that $|h|_S^{-1}(0)|, |h|_S^{-1}(1)| \leq |S| = k - 1$. This means that $h|_{Z \setminus S}$ is not a constant function. As a result, there must be another function $h': Z \to \{0, 1\}$ such that $h'|_S = h|_S$ and $|h'|_{Z \setminus S}^{-1}(0)| = |h|_{Z \setminus S}^{-1}(0)|$. The latter implies that $h' \in \cH$. Since $h'|_S = h|_S$, $S$ cannot be a teaching set for $h$ w.r.t. $\cH$.
\item Consider any set $S \subseteq Z$ of size $k$. If $h|_S \ne \bone$, then we have  $|h_S^{-1}(0)| \leq k$ and $|h_S^{-1}(1)| \leq k - 1$. Similar to above, this means that $h|_{Z \setminus S}$ is not a constant function and, thus, we can find a different $h' \in \cH$ such that $h'|_S = h|_S$. This means that $S$ cannot be a teaching set of $h$. 
\item Consider any set $S \subseteq Z$ of size $k$. If $h|_S \ne \bone$, then the previous point implies that $S$ cannot be a teaching set for $h$ w.r.t. $\cH$. On the other hand, if $h|_S = \bone$, then we have $h|_S = \bone|_S$, which also implies that $S$ cannot be a teaching set for $h$ w.r.t. $\cH \cup \{\bone\}$. As a result, we must have $\TS(h, \cH \cup \{\bone\}) > k$. \qedhere
\end{enumerate}
\end{proof}

We are now ready to present our reduction and prove our main theorem (\Cref{thm:main-red}). Note that the notations for handling tuples and sets of tuples were given at the beginning of \Cref{sec:prelim}. We also remark that $Z \times V$ in the construction presented below corresponds to the set $Y$ in Step II of the overview, whereas $V \times Z$ represents the replication of $V$ in Step III of the overview.

\begin{proof}[Proof of \Cref{thm:main-red}]
Let $G = (V = \{v_1, \dots, v_N\}, E)$ be an input instance to $k$-Dominating Set. Let $\cH = \{h_1, \dots, h_q\}$ be the concept class over domain $Z$ from \Cref{lem:gadget}.

We create the concept class $\cC$ as follows:
\begin{itemize}
\item Let the domain set $X$ be $(V \times Z) \cup (Z \times V)$.
\item For every $h \in \cH$, create a concept $c_h$ where $c_h(Z \times V) = \bone$ and 
\begin{align*}
c_h((v, z)) = h(z) & &\forall (v, z) \in V \times Z.
\end{align*}
\item For every $u \in V$ and $h \in \cH$, create a concept $c_{u, h}$ such that 
\begin{align*}
c_{u, h}((z, v)) = \ind[h(z)=1~\wedge~u=v] & &\forall (z, v) \in Z \times V,    
\end{align*}
and
\begin{align*}
c_{u, h}((v, z)) = \neg \ind[v \text{ dominates } u] & &\forall (v, z) \in V \times Z.
\end{align*} 
\end{itemize}

Note that $\cC$ consists of $n := q(N + 1)$ concepts.
Since $p \leq O(k), q \leq 2^{O(k)}$ and $\cH$ can be constructed in $(pq)^{O(1)}$ time, we can also conclude that $\cC$ can be constructed in $(pqN)^{O(1)} = 2^{O(k)}N^{O(1)}$ time as desired. We will next argue the completeness and soundness of the reduction. For convenience, let $\cC_{\cH} := \{c_h\}_{h \in \cH}$ and $\cC_u := \{c_{u, h}\}_{h \in \cH}$ for all $u \in V$. 
It would also be useful to note a couple of observations, which follow immediately from the definitions of $c_h, c_{u, h}$. (Note that concepts in $\cC_{\cH}$ are the same outside $V \times Z$; and concepts in $\cC_u$ are the same outside $Z \times \{u\}$.) 

\begin{observation} \label{obs:constraint}
For every $h \in \cH$, a set $S \subseteq X$ is a teaching set for $c_h$ w.r.t. $\cC_{\cH}$ iff $\pi_2(S \cap (V \times Z))$ is a teaching set of $h$ w.r.t. $\cH$.
\end{observation}

\begin{observation} \label{obs:vertex}
For every $u \in V$ and $h \in \cH$, a set $S \subseteq X$ is a teaching set for $c_{u, h}$ w.r.t. $\cC_u$ iff $\pi_1(S \cap (Z \times \{u\}))$ is a teaching set of $h$ w.r.t. $\cH$.
\end{observation}

\paragraph{Completeness.} Suppose that there is a $k$-size dominating set $T \subseteq V$. We will show that $\RTD(\cC) \leq k$. By \Cref{lem:rtd-characterization}, this is equivalent to showing that, for every $\cC' \subseteq \cC$, there exists $c \in \cC'$ such that $\TS(c, \cC') \leq k$. To do this, consider the following cases based on whether $\cC' \cap \cC_{\cH} = \emptyset$.
\begin{itemize}
\item Case I: $\cC' \cap \cC_{\cH} \ne \emptyset$. That is, there is $h \in \cH$ such that $c_h \in \cC'$. Let $S_h$ denote a $k$-size teaching set of $h$ w.r.t. $\cH$ (guaranteed to exist by \Cref{lem:gadget}). Then, let $S = \nu(T, S_h)$. Note that $|S| = |S_h| = k$. We will show that $S$ is a teaching set for $c_h$ w.r.t. $\cC$.

First, for any other $c^*_{h'} \in \cC' \cap \cC_{\cH}$, \Cref{obs:constraint} implies that $c^*_{h'}|_S \ne c_h|_S$.

Next, consider any $c_{u, h'} \in \cC' \setminus \cC_{\cH}$. Since $T$ is a dominating set of $V$, there exists $v \in T$ that dominates $u$. This means that $c_{u, h'}((v, z)) = 0$ for all $z \in Z$. By definition of $\nu(T, S_h)$, there must exists $z^v \in Z$ such that $(v, z^v) \in \nu(T, S_h)$. This means that $c_{u, h'}|_S$ cannot be the all-one function. On the other hand, by \Cref{lem:gadget}, we have that $c_h|_S$ is the all-one function. Thus, we have $c_h|_S \ne c_{u, h'}$.

Combining the above two results, $S$ is a teaching set for $c_h$ with respect to $\cC'$.

\item Case II: $\cC' \cap \cC_{\cH} = \emptyset$. In this case, let $c_{u, h}$ be any element of $\cC'$. Again, let $S_h$ denote a $k$-size teaching set of $h$ w.r.t. $\cH$ (guaranteed to exist by \Cref{lem:gadget}). Then, let $S = S_h \times \{u\}$.

Now, consider any $c_{u', h'} \in \cC' \setminus \{c_{u, h}\}$ based on two cases:
\begin{itemize}
\item Case I: $u' \ne u$. By definition, we have $c_{u', h'}|_S$ is the all-zero function. However, \Cref{lem:gadget} implies that $c_{u, h}|_S$ is the all-one function. Therefore, $c_{u, h}|_S \ne c_{u', h'}|_S$.
\item Case II: $u' = u$ and $h' \ne h$. Since $S_h$ is a teaching set for $h$ w.r.t. $\cH$, \Cref{obs:vertex} implies that $c_{u, h}|_S \ne c_{u', h'}|_S$.
\end{itemize}
As a result, we can conclude that $S$ is a teaching set of $c_{u, h}$ with respect to $\cC'$, and $|S| = k$.
\end{itemize}

Thus, we have $\RTD(\cC) \leq k$ as desired.

\paragraph{Soundness.}
Suppose contrapositively that $\RTD(\cC) \leq k$. By \Cref{lem:rtd-characterization}, this means that there exists $c \in \cC$ such that $\TS(c, \cC) \leq k$. Consider the following two cases based on whether $c \in \cC_{\cH}$.
\begin{itemize}
\item Case I: $c \notin \cC_{\cH}$. That is, $c = c_{u, h}$ for some $u \in V, h \in \cH$. Let $S$ be a $k$-size teaching set of $c_{u, h}$ w.r.t. $\cC$. From \Cref{obs:vertex}, $\pi_1(S \cap (Z \times \{u\}))$ must be a teaching set for $h$ w.r.t. $\cH$. By \Cref{lem:gadget} and from $|\pi_1(S \cap (Z \times \{u\}))| \leq |S| \leq k$, this implies that $|\pi_1(S \cap (Z \times \{u\}))| = k$ and that $h|_{\pi_1(S \cap (Z \times \{u\}))} = \bone$. The former means that $S \subseteq Z \times \{u\}$ and the latter means that $c_{u, h}|_S = \bone$. On the other hand, $c_h|_{Z \times \{u\}}$ is also the all-one function by definition. This contradicts with the fact that $S$ is a teaching set for $c_{u, h}$ w.r.t. $\cC$. Thus, this case cannot occur.
\item Case II: $c \in \cC_{\cH}$. That is, $c = c_h$ for some $h \in \cH$. Let $S$ be a $k$-size teaching set of $c_h$ w.r.t. $\cC$. This also implies that $S$ is a teaching set of $c_h$ w.r.t. $\cC_{\cH}$. \Cref{obs:constraint} means that $\pi_2(S \cap (V \times Z))$ must be a teaching set for $h$ w.r.t. $\cH$. By \Cref{lem:gadget} and from $|\pi_2(S \cap (V \times Z))| \leq |S| = k$, this implies that $|\pi_2(S \cap (V \times Z))| = k$ and that $h|_{\pi_2(S \cap (V \times Z))} = \bone$. The former means that $S \subseteq (V \times Z)$ and the latter means that $c_h|_S = \bone$.

Let $T = \pi_1(S)$. Since $|S| \leq k$, we also have $|T| \leq k$. We will next argue that $T$ is a dominating set of $G$. To see this, consider any vertex $u \in V$. Since $c_h|_S = \bone$ and $S$ is a teaching set of $c_h$ w.r.t. $\cC$, there must be $(v, z) \in S$ such that $c_{u, h}((v, z)) = 0$. By definition, however, this means that $v$ dominates $u$. Since $v \in T$, we can conclude that $T$ is a dominating set of size at most $k$ as desired.
\end{itemize}

Hence, if $\RTD(\cC) \leq k$, there must be a dominating set in $G$ of size at most $k$. This concludes our proof.
\end{proof}

\section{Conclusion and Open Questions}
\label{sec:conclusion}

We show that it is unlikely that a polynomial-time algorithm exists for computing the recursive teaching dimension of a concept class that is given explicitly as  input. Another model studied in literature is the ``implicit'' setting, where an input is a circuit that computes $c(x)$ when feeds in (indices of) $c$ and $x$ as inputs. Strong lower bounds are known for VC Dimension and Littlestone's Dimension in this model~\cite{Schaefer99,Schaefer00,MosselU02}. Obtaining one for RTD is an interesting direction.

It would also be interesting to extend the hardness to rule out approximation algorithms for RTD as well. Our reduction is \emph{not} approximation-preserving. If one can come up with an approximation-preserving reduction, then applying hardness of approximating $k$-Dominating Set from the parameterized complexity literature (e.g.~\cite{SLM19}) would yield inapproximability results for RTD.

We remark that our reduction (\Cref{thm:main-red}) also implies that $k$-RTD is hard for the class \emph{LOGNP}, defined by Papadimitriou and Yannakakis~\cite{PapadimitriouY96} since they showed that $\lfloor \log N \rfloor$-Dominating Set is LOGNP-hard. We do not know if the $k$-RTD problem belongs to LOGNP; proving such a containment--or identify the class for which it is complete--is an intriguing open question. The same problem remains open for Littlestone's Dimension~\cite{FrancesL98}.

Similar to the above, our reduction implies that $k$-RTD, when parameterized by $k$, is hard for the class W[2] (the second level of the W-hierarchy\footnote{A more detailed formulation of W-hierarchy and FPT reductions can be found e.g. in \cite{DowneyF13}.}) as $k$-Dominating Set is W[2]-complete~\cite{DowneyF95}. Again, we do not know whether $k$-RTD belongs to W[2].

Finally, there are still other learning-theoretic quantities whose computational complexity is not yet understood. For example, the self-directed learning complexity~\cite{GoldmanRS93,GoldmanS94} of an explicitly given concept class can be computed in $n^{O(\log n)}$ time and therefore, similar to RTD, is unlikely to be NP-hard. Ben-David and Eiron~\cite{Ben-DavidE98} asked whether it can be computed in polynomial time. To the best of our knowledge, this question remains open til this day.

\bibliographystyle{alpha}
\bibliography{ref}

\appendix

\section{Proof of \Cref{thm:log-domset-lb}}
\label{sec:proof-domset}

To prove \Cref{thm:log-domset-lb}, we require a reduction from 3SAT to $k$-Dominating Set. There are multiple such reductions that work, we use the one by P{\u{a}}tra{\c{s}}cu and Williams~\cite{PatrascuW10}, summarized below:

\begin{lemma}[\cite{PatrascuW10}]
For any $k \in \N$, there is an $N^{O(1)}$-time reduction from $n$-variable $m$-clause 3SAT to $k$-Dominating Set on $N$-vertex graph where $N = k \cdot 2^{\lceil n / k \rceil} + m + k$.
\end{lemma}

\begin{proof}[Proof of \Cref{thm:log-domset-lb}]
We use the above reduction with $k = \lceil \sqrt{n} \rceil$. We have\footnote{Note that we may assume w.l.o.g. that $m \leq O(n^3)$.} $N = 2^{\Theta(\sqrt{n})}$ and therefore $k = O(\log N)$. Thus, if there is an $N^{o(\log N)}$-time algorithm for $k$-Dominating Set, then we can run this algorithm to solve 3SAT in time $N^{o(\log N)} = 2^{o(n)}$, violating ETH.
\end{proof}

\end{document}